\documentclass[a4paper,11pt]{article}
\usepackage{bm,mathrsfs,amsmath,amsthm,amssymb,amscd,longtable,array,graphicx,bm}
\newcommand{\nc}{\newcommand}
\nc{\rnc}{\renewcommand}
\nc{\nn}{\nonumber}
\nc{\der}{{\partial}}
\rnc{\Im}{{\rm{Im}\,}}
\rnc{\Re}{{\rm{Re}\,}}
\nc{\db}{\displaybreak[0]\\}
\nc{\bra}{\langle}
\nc{\ket}{\rangle}
\nc{\bs}{\boldsymbol}

\newtheorem{theorem}{Theorem}[section]
\newtheorem{lemma}[theorem]{Lemma}

\newtheorem{proposition}[theorem]{Proposition}

\theoremstyle{definition}

\numberwithin{equation}{section}

\numberwithin{equation}{section}

\begin{document}%
%
\title{
Factorization of rational six vertex model partition functions
}

\author{
Kohei Motegi \thanks{{\it E-mail address}: kmoteg0@kaiyodai.ac.jp}
\\\\
{\it Faculty of Marine Technology, Tokyo University of Marine Science and Technology,}\\
 {\it Etchujima 2-1-6, Koto-Ku, Tokyo, 135-8533, Japan} \\
\\\\
\\
}

\date{\today}

\maketitle

\begin{abstract}
We show factorization formulas for a class of partition functions of rational six vertex model.
First we show factorization formulas for partition functions under triangular boundary.
Further,
by combining the factorization formulas with
the explicit forms of the generalized domain wall boundary partition functions
by Belliard-Pimenta-Slavnov, we derive factorization formulas
for partition functions under trapezoid boundary
which can be viewed as a generalization of triangular boundary.
We also discuss an application to emptiness formation probabilities
under trapezoid boundary which admit determinant representations.
\end{abstract}

\section{Introduction}
Domain wall boundary partition functions \cite{Izergin,Korepin}
is one of the most fundamental classes of partition functions in six vertex models 
\cite{Lieb,Sutherland,Baxter}
which has been deeply investigated and applied to other areas of mathematics.
See \cite{Kuperberg,KZJ,CP,BF} for seminal works.
The fundamental result is that the partition functions can be represented by determinants
called Izergin-Korepin determinants, which the rational version appeared in a different context in \cite{Gaudin}.
Partial domain wall boundary partition functions \cite{GSV,Kostov,FW}
is a generalized class of the ordinary domain wall boundary partition functions
which one side of the boundaries is not fixed to be all spins up (or all spins down), but a mixture of up and down spins.
Recently, for the rational six-vertex model,
a further generalized class which every side of the boundaries is a mixture of up and down spins
was introduced, and the explicit forms representing the generalized domain wall boundary partition functions
was derived by Belliard-Pimenta-Slavnov \cite{BPS}.
The explicit forms also appear in their previous studies on the XXX chain \cite{BSV,BS}.
Two of the explicit expressions in \cite{BPS} are rational version of special functions
of trigonometric type introduced in \cite{KN,MN} in the study of Macdonald operators,
and two determinant forms in \cite{BPS} can be found in \cite{GZZ} which were used to derive the quantum-classical duality.
Other types of determinant forms can be found for example in \cite{Izergin,Gaudin,GSV,Kostov,FW,BS,MPT}
which a systematic understanding was given in \cite{MO}.

In this paper, motivated by the generalized domain wall boundary partition functions,
we introduce and investigate a class of rational six vertex model partition functions under triangular boundary and its
extension to trapezoid boundary.
The shapes of partition functions under triangular boundary and trapezoid boundary which we investigate
look similar to the ones first introduced and studied in \cite{Kuperberg}
and further investigated in \cite{BFK,ABF,BFKtwo,GdGMW},
but the boundary conditions are different.
In their setting, 
besides using the trigonometric $R$-matrix
\cite{FRT,Drinfeld,Jimbo} instead of the rational $R$-matrix,
one boundary uses the off-diagonal $K$-matrix in \cite{Kuperberg} and
more generic nondiagonal $K$-matrices in \cite{BFK,ABF,BFKtwo,GdGMW}
and the rest of the boundaries are fixed,
whereas in this paper we use identity matrix on one boundary
and the rest of the boundaries are free (each state is a mixture of up and down spins).
This setting for the free boundary parts is basically the same with \cite{BPS}.
The setting for triangular boundary looks also similiar to a slice of partition functions given in
\cite{KMO}, but one boundary is fixed instead of using the identity matrix.
Another difference is that the $R$-matrix is an operator-valued one which
is a $q=0$ version of the one introduced in \cite{BaSe}.
The operator-valued $R$-matrix is equivalent to the
three-dimensional $R$-matrix which satisfies the tetrahedron equation
rather than the Yang-Baxter equation.

We show that the rational six vertex model partition functions
which we introduce in this paper have factorized forms.
First we derive the explicit forms for the case of triangular boundary
by using unitarity relation and also showing several specializations vanish.
Next, combining the explicit forms for triangular boundary and the generalized domain wall boundary partition functions
by Belliard-Pimenta-Slavnov
\cite{BPS}, we derive factorized forms for trapezoid boundary.
We also introduce a version of emptiness formation probability
and investigate by combining the results for partition functions.

This paper is organized as follows.
In section 2, we recall the rational $R$-matrix and its basic properties.
In section 3, we introduce partition functions under triangular boundary
and show their explicit factorized forms.
In section 4, we first recall the results for the generalized domain wall boundary partition functions.
Then we introduce partition functions under trapezoid boundary and derive
their explicit forms.
Finally we discuss an application to a version of emptiness formation probability
which admits determimant representations.

\section{Rational six vertex model}

In this section, we introduce the rational six vertex model and
recall some of its basic properties. 
\begin{figure}[htb] 
\centering
\includegraphics[width=10cm]{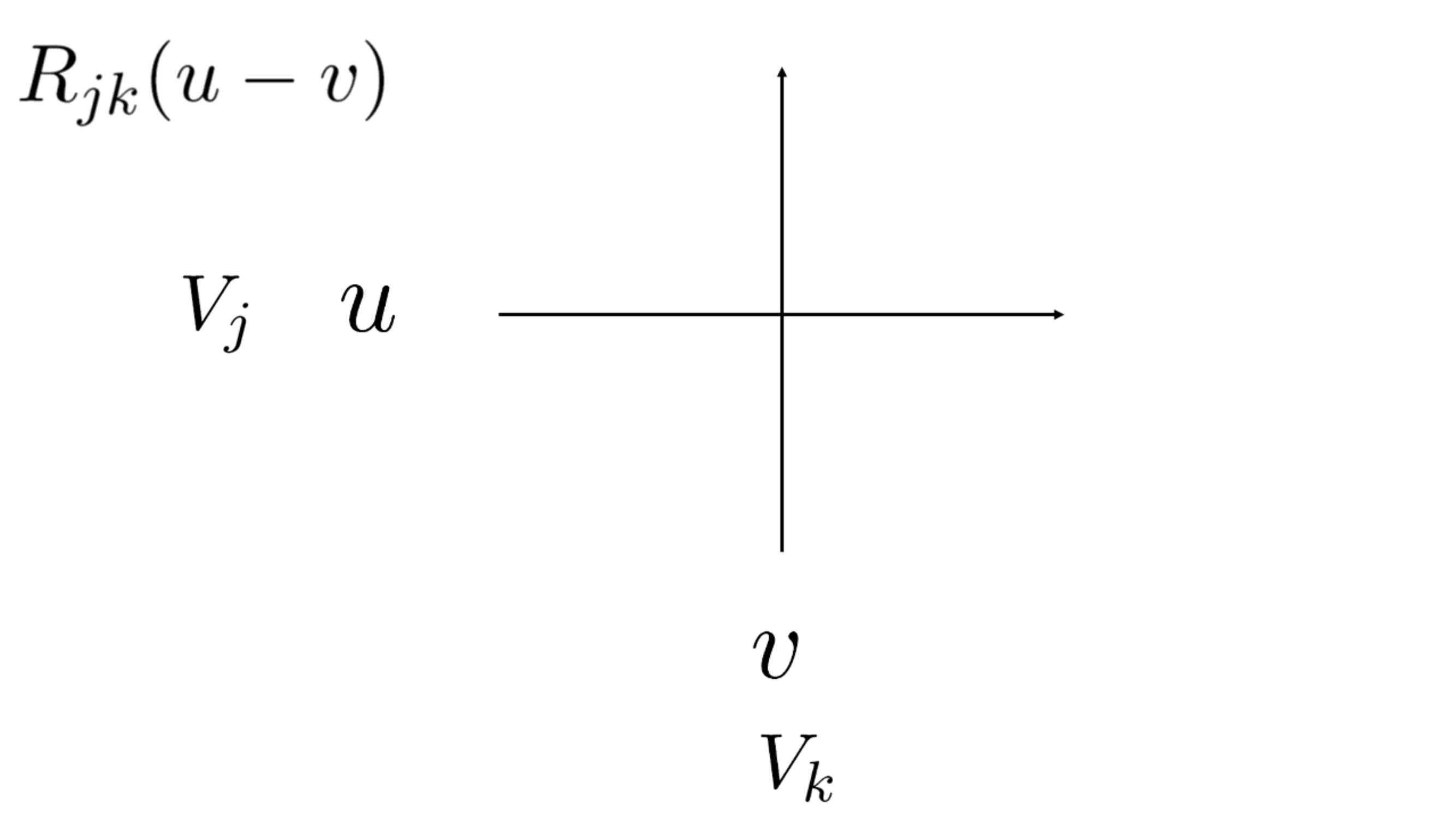}
\caption{The $R$-matrix acting on $V_j \otimes V_k$. The horizontal and vertical line represents $V_j$ and $V_k$
respectively, and to each line variable $u$ and $v$ is associated.
}
\label{figurermatrixzero}
\end{figure}

Let $V$ be a complex two-dimensional space
and $\{ | 1 \rangle, |2 \rangle \}$ be its orthonormal basis.
Let $V^*$ be the dual space of $V$ and denote the dual orthonormal basis as $\{ \langle 1|, \langle 2| \}$.
Using this bra-ket notation, the orthonormality relation is expressed as $\langle k|\ell \rangle=\delta_{k \ell}$ for $k,\ell=1,2$
where $\delta_{k \ell}$ is the Kronecker delta.
For arbitrary complex numbers $n_1,n_2,e_1,e_2,s_1,s_2,w_1,w_2$,
define four vectors $|n \rangle=n_1 | 1 \rangle+n_2 | 2 \rangle$,
$|e \rangle=e_1 | 1 \rangle+e_2 | 2 \rangle$, $| s \rangle =s_1 | 1 \rangle+s_2 |2 \rangle$,
$|w \rangle=w_1 | 1 \rangle+w_2 | 2 \rangle$,
and their duals
$\langle n|=n_1 \langle 1|+n_2 \langle 2|$, $\langle e|=e_1 \langle 1|+e_2 \langle 2|$, $\langle s|=s_1 \langle 1|+s_2 \langle 2|$, $\langle w|=w_1 \langle 1|+w_2 \langle 2|$.
To distinguish different two-dimensional spaces, we use subscript.
For the space $V_j,$ we denote the corresponding four vectors as $|n \rangle_j$, $|e \rangle_j$, $|s \rangle_j$, $|w \rangle_j$,
and for the dual $V_j^*$, we denote as  ${}_j \langle n|$, ${}_j \langle e|$, ${}_j \langle s |$, ${}_j \langle w|$.

\begin{figure}[h] 
\centering
\includegraphics[width=8cm]{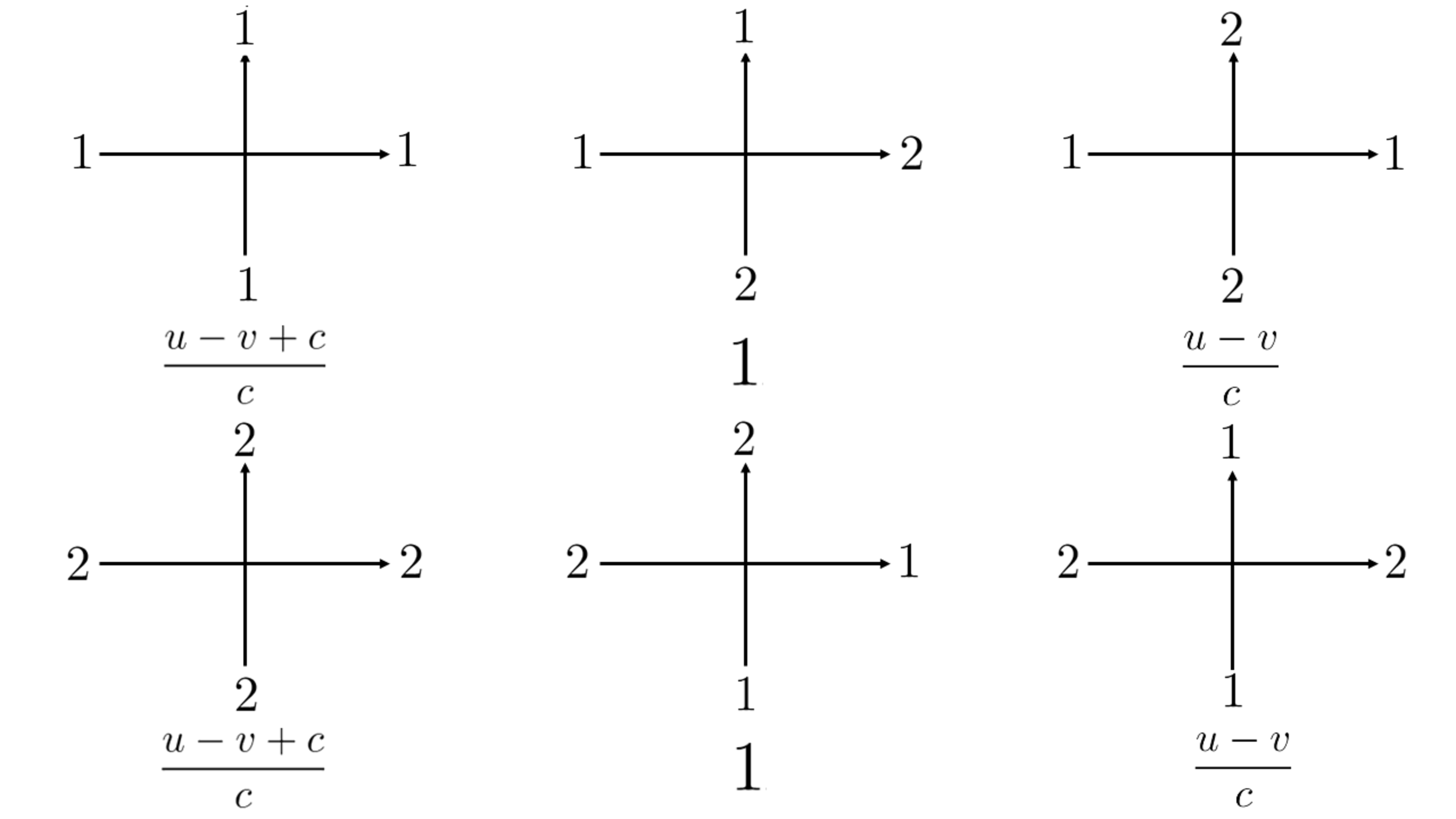}
\caption{Nonzero matrix elements of the rational $R$-matrix.
}
\label{figurermatrixelement}
\end{figure}

The $R$-matrix acting on $V_j \otimes V_k$ can be graphically represented as
Figure \ref{figurermatrixzero}.

Using the bra-ket notation,
the rational $R$-matrix acting on $V_j \otimes V_k$
which we use is given by
\begin{align}
R_{jk}(u-v)=&\frac{u-v+c}{c} 
(| 1 \rangle_j \otimes |1 \rangle_k 
 {}_j \langle 1| \otimes {}_k \langle 1|
+| 2 \rangle_j \otimes |2 \rangle_k 
 {}_j \langle 2| \otimes {}_k \langle 2|
 ) \nn \\
&+\frac{u-v}{c} 
(| 1 \rangle_j \otimes |2 \rangle_k 
 {}_j \langle 1| \otimes {}_k \langle 2|
+| 2 \rangle_j \otimes |1 \rangle_k 
 {}_j \langle 2| \otimes {}_k \langle 1|
 ) \nn \\
&+| 1 \rangle_j \otimes |2 \rangle_k 
 {}_j \langle 2| \otimes {}_k \langle 1|
+| 2 \rangle_j \otimes |1 \rangle_k 
 {}_j \langle 1| \otimes {}_k \langle 2|. \label{twociteR}
\end{align}
Note also the fundamental fact that when $u=v$, the $R$-matrix reduces to permutation operator $R_{jk}(0)=P_{jk}$
which we frequently use in this paper.
We extend the $R$-matrix $R_{jk}(u-v)$ to the one acting on $V_1 \otimes V_2 \otimes \cdots \otimes V_n$
for $n \ge 3$ by acting on $V_j$ and $V_k$ as \eqref{twociteR} and acting on the other spaces
as identity.

\begin{figure}[tbh] 
\centering
\includegraphics[width=10cm]{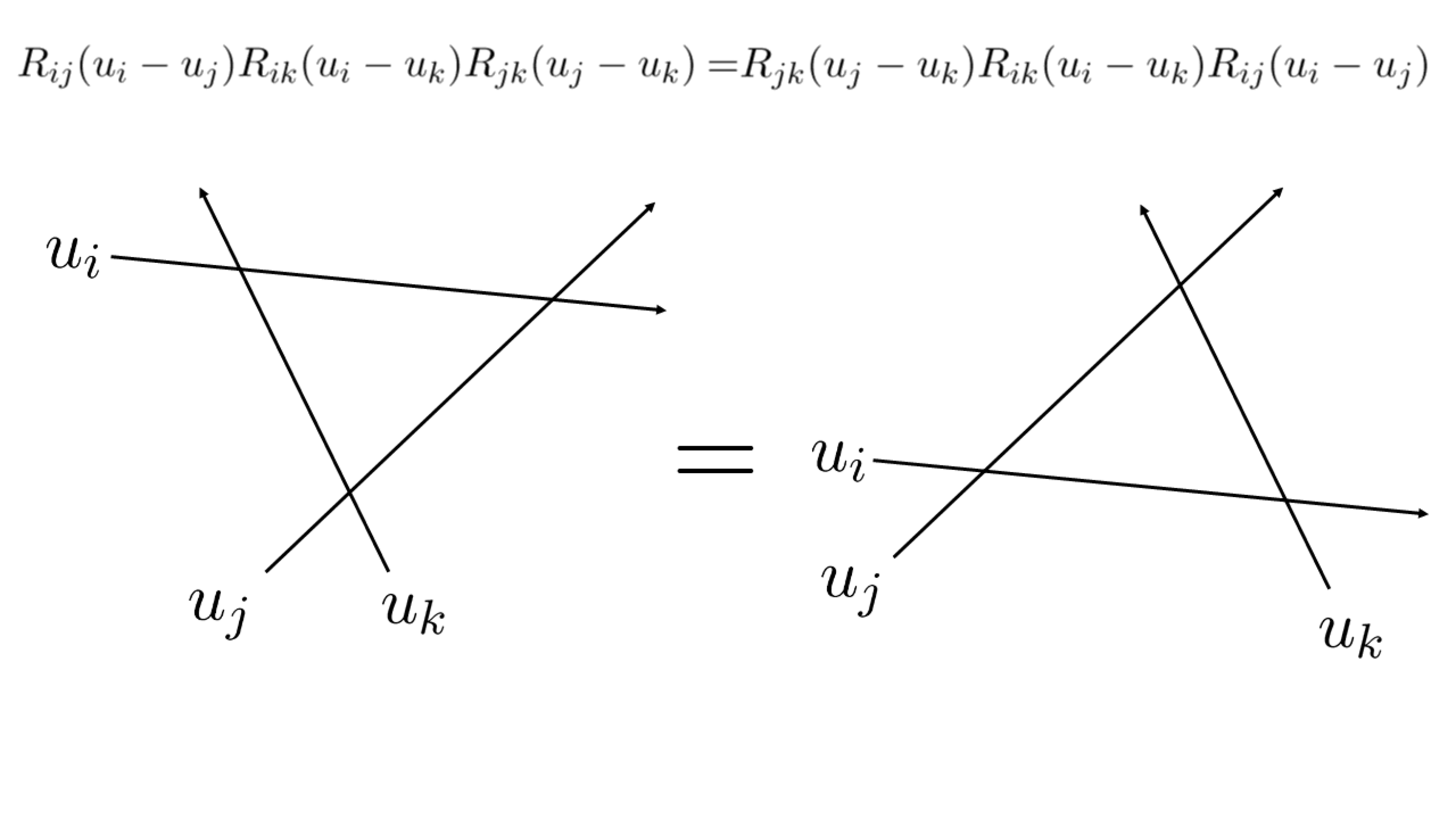}
\caption{Yang-Baxter relation \eqref{YBE}.}
\label{figureYBE}
\end{figure}

\begin{figure}[tbh] 
\centering
\includegraphics[width=8cm]{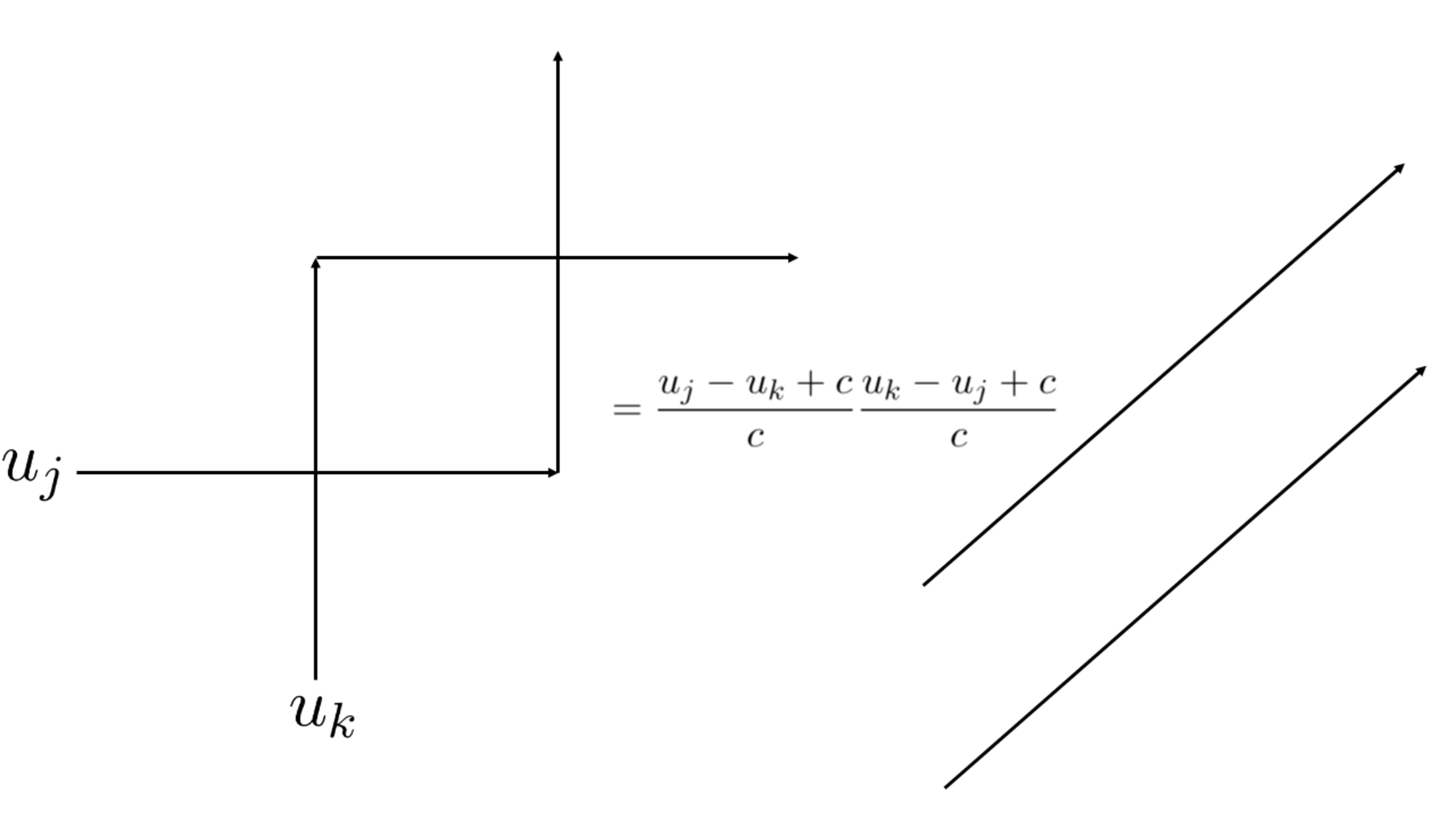}
\caption{Unitarity relation \eqref{unitarity}.}
\label{figureunitarity}
\end{figure}

See Figure \ref{figurermatrixelement}
where for each six configuration, the corresponding weight associated
is denoted below. For the remaining ten configurations which are not drawn, the weights associated are all zero.
Note this graphical description is slightly different from \cite[Figure 7]{BPS}.

The $R$-matrix satisfies the
Yang-Baxter relation (Figure \ref{figureYBE})
\begin{align}
&R_{ij}(u_i-u_j)R_{ik}(u_i-u_k)R_{jk}(u_j-u_k) \nn \\
=&R_{jk}(u_j-u_k)R_{ik}(u_i-u_k)R_{ij}(u_i-u_j) \in \mathrm{End}(V_i \otimes V_j \otimes V_k),
\label{YBE}
\end{align}
and the unitarity relation (Figure \ref{figureunitarity})
\begin{align}
R_{jk}(u_j-u_k)R_{kj}(u_k-u_j)=\frac{u_j-u_k+c}{c} \frac{u_k-u_j+c}{c} \mathrm{id} \in \mathrm{End}(V_j \otimes V_k).
\label{unitarity}
\end{align}

\section{Factorization of partition functions under triangular boundary}

We introduce the following partition functions
under triangular boundary
(Figure \ref{figurehaｌftwist})
\begin{align}
&Z_n(u_1,u_2,\dots,u_n) \nn \\
=&{}_1 \langle e| \otimes \cdots \otimes {}_n \langle e|
R_{n-1,n}(u_{n-1}-u_n)
(R_{n-2,n}(u_{n-2}-u_n)R_{n-2,n-1}(u_{n-2}-u_{n-1})) \nn \\
&\times \cdots \times
(R_{2n}(u_2-u_n) \cdots R_{23}(u_2-u_3))
(R_{1n}(u_1-u_n) \cdots R_{12}(u_1-u_2))
|s \rangle_1 \otimes \cdots \otimes |s \rangle_n. \label{sumhalftwist}
\end{align}

\begin{figure}[h] 
\centering
\includegraphics[width=10cm]{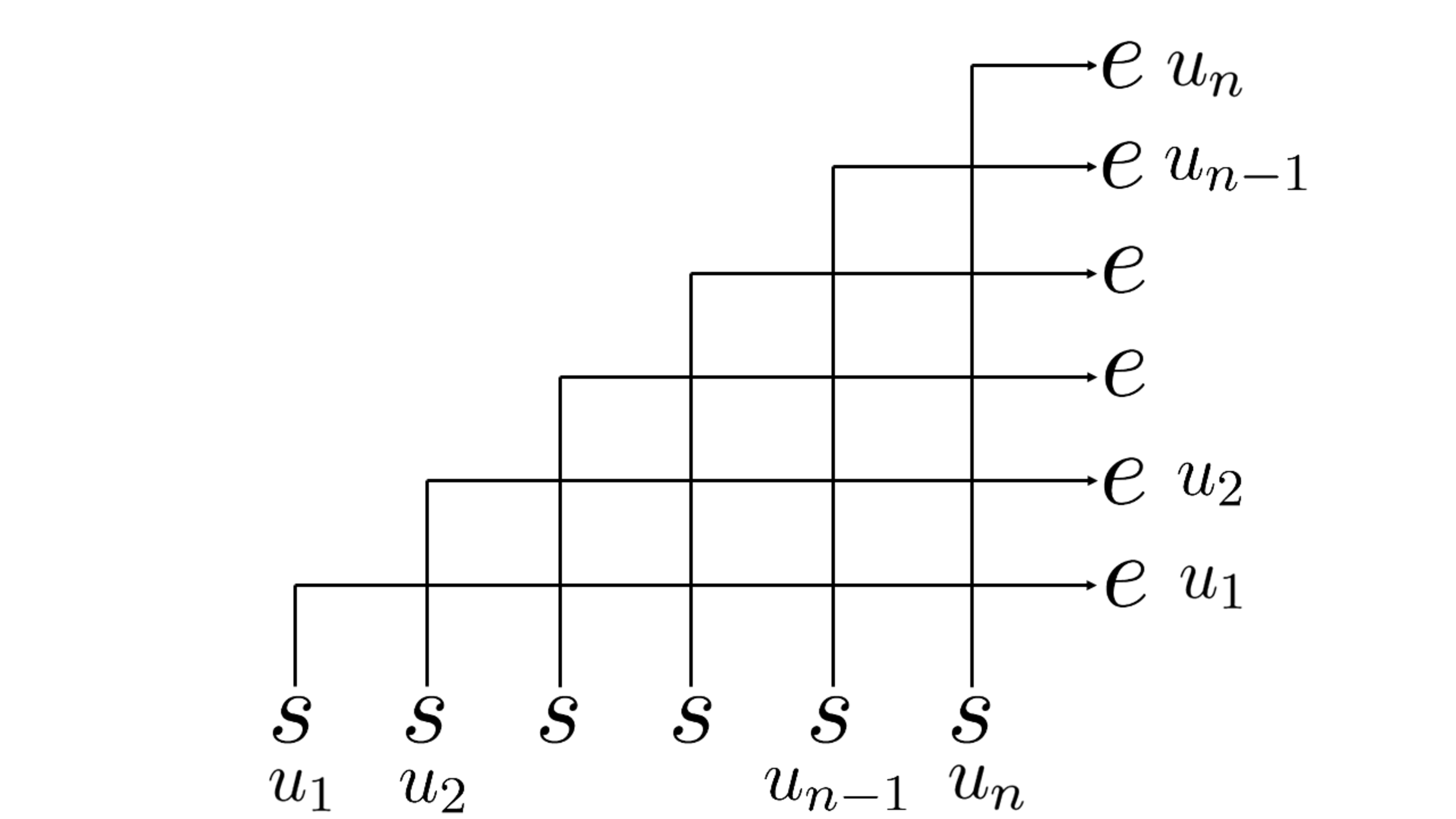}
\caption{Partition functions
under triangular boundary $Z_n(u_1,u_2,\dots,u_n)$ \eqref{sumhalftwist}.
At the bottom boundary, each state is  a mixture of $|1 \rangle$ and $|2 \rangle$ given by
$|s \rangle=s_1|1 \rangle+s_2|2 \rangle$.
At the right boundary, each state is given by
$\langle e |=e_1 \langle 1 |+e_2  \langle 2 |$.
}
\label{figurehaｌftwist}
\end{figure}

\begin{figure}[h] 
\centering
\includegraphics[width=10cm]{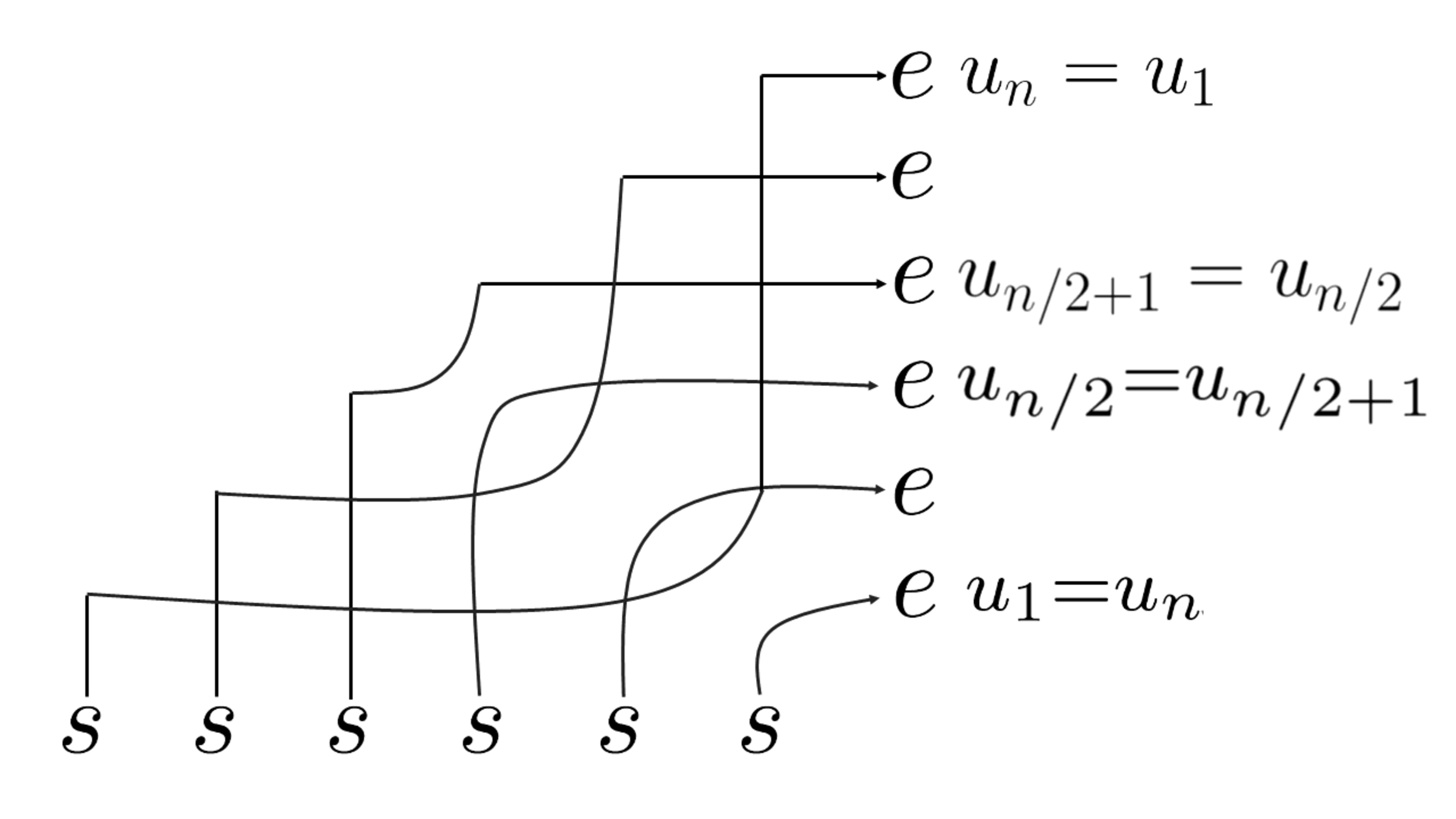}
\caption{
Specialization
$u_1=u_n, u_2=u_{n-1}, \dots,u_{n/2}=u_{n/2+1}$ of the partition functions
$Z_n(u_1,u_2,\dots,u_n)$ for $n$ even.
Diagonal part of the $R$-matrices are turned into permutation operators.
}
\label{figurespecializtionhalftwist}
\end{figure}

\begin{figure}[h] 
\centering
\includegraphics[width=10cm]{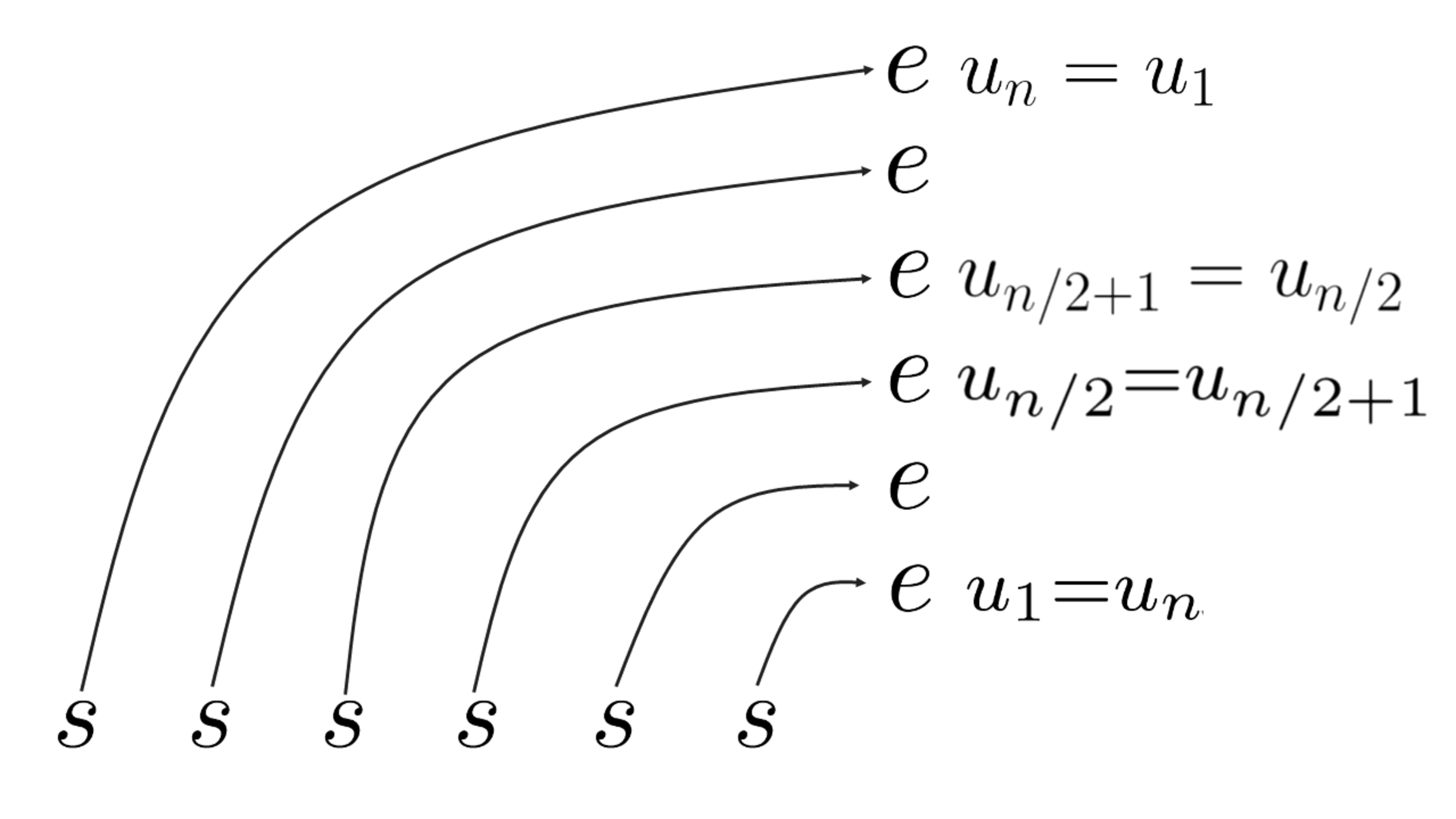}
\caption{
Using the unitarity relation repeatedly,
Figure \ref{figurespecializtionhalftwist} for
$Z_n(u_1,u_2,\dots,u_n)$ specialized at
$u_1=u_n, u_2=u_{n-1}, \dots,u_{n/2}=u_{n/2+1}$
can be reduced to the following figure up to overall factors coming from the unitarity relation.
Each line represents $e_1 s_1+e_2 s_2$.}
\label{figureunwinded}
\end{figure}

We show the following factorized forms for the partition functions
$Z_n(u_1,u_2,\dots,u_n)$.
\begin{proposition}
We have
\begin{align}
Z_n(u_1,u_2,\dots,u_n)
=(e_1 s_1+e_2 s_2)^n \prod_{1 \le i < j \le n} \Bigg( \frac{u_i-u_j+c}{c} \Bigg).
\label{factorizationhalftwist}
\end{align}
\end{proposition}
\begin{proof}
To prove \eqref{factorizationhalftwist},
it is enough to show the following properties. \\
(i) $Z_n(u_1,u_2,\dots,u_n)$ is a polynomial of degree at most $n-1$ in $u_j$, $j=1,\dots,n$. \\
(ii) $Z_n(u_1,u_2,\dots,u_n)$ satisfies
\begin{align}
Z_n(u_1,u_2,\dots,u_n)|_{u_n=u_j+c}=0, \ \ \ j=1,\dots,n-1. \label{vanishing}
\end{align}
(iii) $Z_n(u_1,u_2,\dots,u_n)$ satisfies
\begin{align}
&Z_n(u_1,u_2,\dots,u_n)|_{u_1=u_n, u_2=u_{n-1}, \dots,u_{n/2}=u_{n/2+1}} \nn \\
=&(e_1 s_1+e_2 s_2)^n
\prod_{1 \le j < k \le n-j}
\Bigg( \frac{u_j-u_k+c}{c} \frac{u_k-u_j+c}{c} \Bigg)
 \Bigg|_{u_1=u_n, u_2=u_{n-1}, \dots,u_{n/2}=u_{n/2+1}}, \label{specializationone}
\end{align}
for $n$ even, and
\begin{align}
&Z_n(u_1,u_2,\dots,u_n)|_{u_1=u_n, u_2=u_{n-1}, \dots,u_{(n-1)/2}=u_{(n+1)/2}} \nn \\
=&(e_1 s_1+e_2 s_2)^n
\prod_{1 \le j < k \le n-j}
\Bigg( \frac{u_j-u_k+c}{c} \frac{u_k-u_j+c}{c} \Bigg)
 \Bigg|_{u_1=u_n, u_2=u_{n-1}, \dots,u_{(n-1)/2}=u_{(n+1)/2}}, \label{specilizationtwo}
\end{align}
for $n$ odd.

Let us first show (i).
By the definition of $Z_n(u_1,u_2,\dots,u_n)$,
the variable $u_j$ comes from $(n-1)$ $R$-matrices $R_{1,j}(u_{1}-u_j),\dots,R_{j-1,j}(u_{j-1}-u_j),R_{j,j+1}(u_j-u_{j+1}),\dots,R_{jn}(u_j-u_n)$.
Since each $R$-matrix element contributes at most degree 1 to $u_j$, the claim follows.

Next we show (ii). The case $j=n-1$ can be proven by using
\begin{align}
R_{n-1,n}(u_{n-1}-u_n)|_{u_n=u_{n-1}+c}=&
(|2 \rangle_{n-1} \otimes |1 \rangle_n-|1 \rangle_{n-1} \otimes |2 \rangle_n)
 {}_{n-1} \langle 1| \otimes {}_n \langle 2|
\nn \\
&
+(|1 \rangle_{n-1} \otimes |2 \rangle_n
-|2 \rangle_{n-1} \otimes |1 \rangle_n){}_{n-1} \langle 2| \otimes {}_n \langle 1|,
\end{align}
to decompose \eqref{sumhalftwist} evaluated at $u_n=u_{n-1}+c$
as
\begin{align}
&({}_1 \langle e| \otimes \cdots \otimes 
{}_{n-2} \langle e|) \otimes
{}_{n-1} \langle e| \otimes
{}_n \langle e|
(|2 \rangle_{n-1} \otimes |1 \rangle_n-|1 \rangle_{n-1} \otimes |2 \rangle_n) \nn \\
&\times  {}_{n-1} \langle 1| \otimes {}_n \langle 2| 
(R_{n-2,n}(u_{n-2}-u_n)R_{n-2,n-1}(u_{n-2}-u_{n-1})) \nn \\
&\times \cdots \times
(R_{1n}(u_1-u_{n-1}-c) \cdots R_{12}(u_1-u_2))
|s \rangle_1 \otimes \cdots \otimes |s \rangle_n \nn \\
&+({}_1 \langle e| \otimes \cdots \otimes 
{}_{n-2} \langle e|) \otimes
{}_{n-1} \langle e| \otimes
{}_n \langle e|
(|1 \rangle_{n-1} \otimes |2 \rangle_n-|2 \rangle_{n-1} \otimes |1 \rangle_n) \nn \\
&\times  {}_{n-1} \langle 2| \otimes {}_n \langle 1| 
(R_{n-2,n}(u_{n-2}-u_n)R_{n-2,n-1}(u_{n-2}-u_{n-1})) \nn \\
&\times \cdots \times
(R_{1n}(u_1-u_{n-1}-c) \cdots R_{12}(u_1-u_2))
|s \rangle_1 \otimes \cdots \otimes |s \rangle_n \nn \\
=&(e_2 e_1-e_1 e_2)
{}_1 \langle e| \otimes \cdots \otimes 
{}_{n-2} \langle e| \otimes  {}_{n-1} \langle 1| \otimes {}_n \langle 2| 
(R_{n-2,n}(u_{n-2}-u_n)R_{n-2,n-1}(u_{n-2}-u_{n-1})) \nn \\
&\times \cdots \times
(R_{1n}(u_1-u_{n-1}-c) \cdots R_{12}(u_1-u_2))
|s \rangle_1 \otimes \cdots \otimes |s \rangle_n \nn \\
&+(e_1 e_2-e_2 e_1)
{}_1 \langle e| \otimes \cdots \otimes 
{}_{n-2} \langle e|
\otimes  {}_{n-1} \langle 2| \otimes {}_n \langle 1| 
(R_{n-2,n}(u_{n-2}-u_n)R_{n-2,n-1}(u_{n-2}-u_{n-1})) \nn \\
&\times \cdots \times
(R_{1n}(u_1-u_{n-1}-c) \cdots R_{12}(u_1-u_2))
|s \rangle_1 \otimes \cdots \otimes |s \rangle_n \nn \\
=&0.
\end{align}
For the other cases $u_n=u_j+c, \ j=1,\dots,n-2$, we use the following trace identity \cite[(28)]{BPS}:
\begin{align}
\mathrm{Tr}_a (|s \rangle_a {}_a \langle e | X_a)={}_a \langle e|X_a |s \rangle_a,
\label{traceidentity}
\end{align}
for any matrix $X$ acting on $V_a$.
One can show the other cases by
using \eqref{traceidentity}
and rewriting \eqref{sumhalftwist} in the following way
\begin{align}
&Z_n(u_1,u_2,\dots,u_n) \nn \\
=&\mathrm{Tr}_{12\dots n}(
(|s \rangle_1
{}_1 \langle e|) \otimes \cdots \otimes (|s \rangle_n {}_n \langle e|)
R_{n-1,n}(u_{n-1}-u_n)
(R_{n-2,n}(u_{n-2}-u_n)R_{n-2,n-1}(u_{n-2}-u_{n-1})) \nn \\
&\times \cdots \times
(R_{2n}(u_2-u_n) \cdots R_{23}(u_2-u_3))
(R_{1n}(u_1-u_n) \cdots R_{12}(u_1-u_2))
) \nn \\
=&\mathrm{Tr}_{12\dots n} ( R_{n-1,n}(u_{n-1}-u_n)
\cdots (R_{j+1,n}(u_{j+1}-u_n) \cdots R_{j+1,j+2}(u_{j+1}-u_{j+2})) \nn \\
&\times (|s \rangle_1
{}_1 \langle e|) \otimes \cdots \otimes (|s \rangle_n {}_n \langle e|)
(R_{j,n}(u_{j}-u_n) \cdots R_{j,j+1}(u_{j}-u_{j+1})) \nn \\
&\times \cdots \times (R_{1n}(u_1-u_n) \cdots R_{12}(u_1-u_2))
) \nn \\
=&\mathrm{Tr}_{12\dots n} (  (|s \rangle_1
{}_1 \langle e|) \otimes \cdots \otimes (|s \rangle_n {}_n \langle e|)
(R_{j,n}(u_{j}-u_n) \cdots R_{j,j+1}(u_{j}-u_{j+1})) \nn \\
&\times \cdots \times (R_{1n}(u_1-u_n) \cdots R_{12}(u_1-u_2)) R_{n-1,n}(u_{n-1}-u_n) \nn \\
&\times \cdots \times (R_{j+1,n}(u_{j+1}-u_n) \cdots R_{j+1,j+2}(u_{j+1}-u_{j+2}) )
) \nn \\
=&{}_1 \langle e| \otimes \cdots \otimes {}_n \langle e|
(R_{j,n}(u_{j}-u_n) \cdots R_{j,j+1}(u_{j}-u_{j+1})) \nn \\
&\times \cdots \times (R_{1n}(u_1-u_n) \cdots R_{12}(u_1-u_2)) R_{n-1,n}(u_{n-1}-u_n) \nn \\
&\times \cdots \times (R_{j+1,n}(u_{j+1}-u_n) \cdots R_{j+1,j+2}(u_{j+1}-u_{j+2}) )
|s \rangle_1 \otimes \cdots \otimes |s \rangle_n. 
\end{align}
Here we used the fact that
$[G \otimes G, R_{jk}(u_j-u_k)]=0$ for arbitrary $G$ and applied to $G=|s \rangle \langle e|$.

To show (iii),
we use the graphical description of $Z_n(u_1,u_2,\dots,u_n)$.
We show the case when $n$ is even.
The case when $n$ is odd can be proved in the same way.
We use the fact that $R_{jk}(u_j-u_k)$ reduces to permutation operator at $u_j=u_k$:
$R_{jk}(u_j-u_k)|_{u_j=u_k}=P_{jk}$.
Then we note $Z_n(u_1,u_2,\dots,u_n)$ specialized to
$u_1=u_n, u_2=u_{n-1}, \dots,u_{n/2}=u_{n/2+1}$
can be graphically represented as Figure \ref{figurespecializtionhalftwist}.
Using the Yang-Baxter relation \eqref{YBE} and the unitarity relation \eqref{unitarity} repeatedly,
Figure \ref{figurespecializtionhalftwist} can be transformed into Figure \ref{figureunwinded}.
There are $n$ lines and each of them gives the factor $e_1 s_1+e_2 s_2$.
Taking into account the factors coming from the unitarity relation, we find
$Z_n(u_1,u_2,\dots,u_n)$ satisfies
\begin{align}
&Z_n(u_1,u_2,\dots,u_n)|_{u_1=u_n, u_2=u_{n-1}, \dots,u_{n/2}=u_{n/2+1}} \nn \\
=&\prod_{1 \le j < k \le n-j}
\Bigg( \frac{u_j-u_k+c}{c} \frac{u_k-u_j+c}{c} \Bigg)
 \Bigg|_{u_1=u_n, u_2=u_{n-1}, \dots,u_{n/2}=u_{n/2+1}} \nn \\
&\times {}_1 \langle e| \otimes \cdots \otimes {}_n \langle e| \mathrm{id} |s \rangle_1 \otimes \cdots \otimes |s \rangle_n
\nn \\
=&(e_1 s_1+e_2 s_2)^n
\prod_{1 \le j < k \le n-j}
\Bigg( \frac{u_j-u_k+c}{c} \frac{u_k-u_j+c}{c} \Bigg)
 \Bigg|_{u_1=u_n, u_2=u_{n-1}, \dots,u_{n/2}=u_{n/2+1}}.
\end{align}

Finally we explain Properties (i), (ii), (iii) give \eqref{factorizationhalftwist}.
From Property (ii) \eqref{vanishing},
we get
\begin{align}
Z_n(u_1,u_2,\dots,u_n)=C \prod_{1 \le i < j \le n}(u_i-u_j+c), \label{propiandii}
\end{align}
where $C$ is some polynomial in $u_1,\dots,u_n$.
From Property (i), we conclude $C$ is independent of $u_1,\dots,u_n$,
since if it contains any one of them, say $u_j$, then the degree of $Z_n(u_1,u_2,\dots,u_n)$
in $u_j$ becomes equal to or larger than $n$, which leads to contradiction.
We can determine the constant $C$ using (iii).
For $n$ even, \eqref{specializationone} can be rewritten as
\begin{align}
&Z_n(u_1,u_2,\dots,u_n)|_{u_1=u_n, u_2=u_{n-1}, \dots,u_{n/2}=u_{n/2+1}} \nn \\
=&(e_1 s_1+e_2 s_2)^n
\prod_{1 \le j < k \le n-j}
\Bigg( \frac{u_j-u_k+c}{c} \frac{u_k-u_{n+1-j}+c}{c} \Bigg)
 \Bigg|_{u_1=u_n, u_2=u_{n-1}, \dots,u_{n/2}=u_{n/2+1}} \nn \\
=&(e_1 s_1+e_2 s_2)^n
\prod_{1 \le j < k \le n-j} \frac{u_j-u_k+c}{c} 
\prod_{1 \le j < k \le n-j} \frac{u_k-u_{n+1-j}+c}{c} 
 \Bigg|_{u_1=u_n, u_2=u_{n-1}, \dots,u_{n/2}=u_{n/2+1}} \nn \\
=&(e_1 s_1+e_2 s_2)^n
\prod_{1 \le j < k \le n-j} \frac{u_j-u_k+c}{c} 
\prod_{\substack{ 1 \le j < k \\ n+1-j \le k}} \frac{u_j-u_{k}+c}{c} 
 \Bigg|_{u_1=u_n, u_2=u_{n-1}, \dots,u_{n/2}=u_{n/2+1}} \nn \\
=&(e_1 s_1+e_2 s_2)^n
\prod_{1 \le j < k \le n} \frac{u_j-u_k+c}{c}
\Bigg|_{u_1=u_n, u_2=u_{n-1}, \dots,u_{n/2}=u_{n/2+1}}. \label{rewritespecialization}
\end{align}
Comparing
\eqref{propiandii} and \eqref{rewritespecialization},
we conclude $C=(e_1 s_1+e_2 s_2)^n/c^{n(n-1)/2}$
and hence \eqref{factorizationhalftwist} follows. The case when $n$ odd can be proved in the same way.
\end{proof}

\section{Factorization of partition functions under trapezoid boundary}

\begin{figure}[h] 
\centering
\includegraphics[width=10cm]{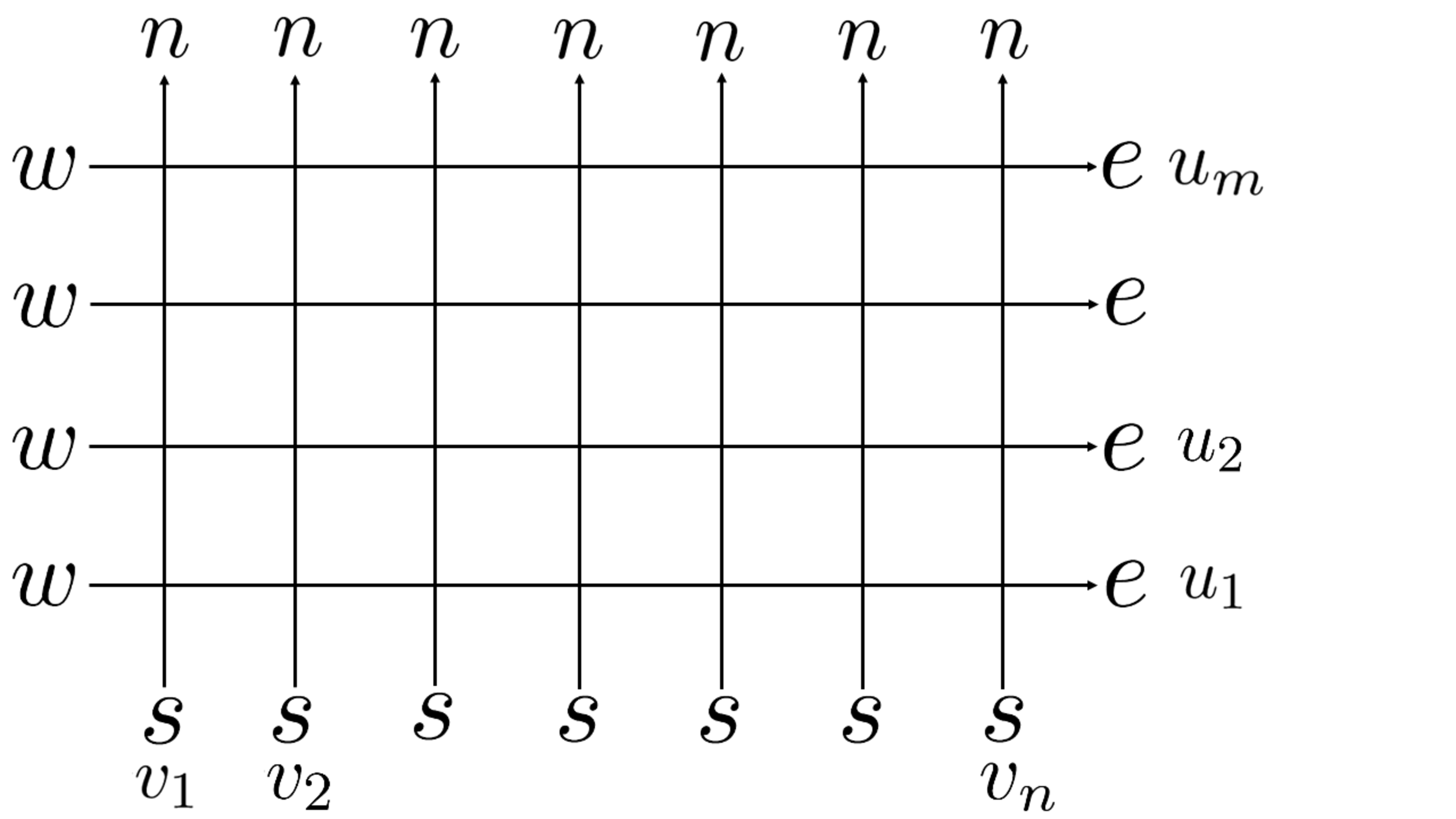}
\caption{The generalized domain wall boundary partition functions
$Z_{m,n}(u_1,\dots,u_m|v_1,\dots,v_n)$ \eqref{GDWBPF}.}
\label{figuredomainwall}
\end{figure}

In this section, we extend the factorization property proved for triangular boundary
to trapezoid boundary.
Before showing this,
we first recall the generalized domain wall boundary partition functions
and some of their explicit forms derived in \cite{BPS}.

The generalized domain wall boundary partition functions introduced
by Belliard-Pimenta-Slavnov are
\begin{align}
&Z_{m,n}(u_1,\dots,u_m|v_1,\dots,v_n) \nn \\
=&{}_{a_m} \langle e| \otimes \cdots \otimes {}_{a_1} \langle e|
\otimes {}_1 \langle n| \otimes \cdots \otimes {}_n \langle n|
(R_{a_m,n}(u_m-v_n) \cdots R_{a_m,1}(u_m-v_1)
) \times \cdots \nn \\
\times&
(R_{a_1,n}(u_1-v_n) \cdots R_{a_1,1}(u_1-v_1)
)
|w \rangle_{a_m} \otimes \cdots \otimes |w \rangle_{a_1} 
\otimes |s\rangle_1 \otimes \cdots \otimes |s \rangle_n,
\label{GDWBPF}
\end{align}
graphically represented as Figure
\ref{figuredomainwall}.
Specializing to $|s \rangle=|1 \rangle$, $|w \rangle=|2 \rangle$, $\langle n|=\langle 2|$, $\langle e|=\langle 1|$
reduces to the ordinary domain wall boundary partition functions.

One of the expressions derived in
\cite{BPS} is the following (writing down more explicitly).

\begin{theorem} (Belliard-Pimenta-Slavnov \cite[(93)]{BPS})
We have
\begin{align}
&Z_{m,n}(u_1,\dots,u_m|v_1,\dots,v_n) \nn \\
=&\frac{(n_1 s_1+n_2 s_2)^{n-m} (e_1 s_1+e_2 s_2)^m (n_1 w_1+n_2 w_2)^m
}{c^{mn}} \nn \\
\times&\sum_{K \subset [1,\dots,m]} (-\beta)^{|K|} \prod_{\substack{ i \in K \\ 1 \le k \le n }}  (u_i-v_k)
\prod_{\substack{ j \not\in K \\ 1 \le k \le n }}  (u_j-v_k+c)
 \prod_{\substack{ i \in K \\ j \not\in K }} \frac{u_i-u_j+c}{u_i-u_j}.
\label{BPSexpression}
\end{align}
Here, $[1,\dots,m]:=\{ j \ | \ j=1,\dots,m \}$ and
for $K$ a subset of $[1,\dots,m]$, we denote the number of elements of $K$ as $|K|$.
The symbol for summation means we take sum over all subsets of $[1,\dots,m].$
$\beta$ is explicitly given by
$\displaystyle \beta=\frac{(e_2 n_1-e_1 n_2)(s_2 w_1-s_1 w_2)}{(e_1 s_1+e_2 s_2)(n_1 w_1+n_2 w_2)}$.
\end{theorem}
Using \cite[Proposition A.4.]{BSV},
\eqref{BPSexpression} can be tranformed into the
 following determinant expression.
\begin{theorem} (Belliard-Pimenta-Slavnov \cite[(91)]{BPS})
We have
\begin{align}
&Z_{m,n}(u_1,\dots,u_m|v_1,\dots,v_n) \nn \\
=&\frac{(n_1 s_1+n_2 s_2)^{n-m} (e_1 s_1+e_2 s_2)^m (n_1 w_1+n_2 w_2)^m 
}{c^{mn}} 
\prod_{\substack{1 \le i \le m \\ 1 \le j \le n}}
(u_i-v_j)
\nn \\
\times& \det_{1 \le j,k \le m}
\Bigg( \delta_{jk} \prod_{i=1}^n \frac{u_j-v_i+c}{u_j-v_i}
-\frac{\beta c}{u_j-u_k+c} \prod_{ \substack{i=1 \\ i \neq j}  }^m \frac{u_j-u_i+c}{u_j-u_i}
\Bigg).
\label{BPSdeterminant}
\end{align}
\end{theorem}
For example, when $m=n=1$, by direct calculation we have
\begin{align}
Z_{1,1}(u|v)=&(e_1 w_1 n_1 s_1+e_2 w_2 n_2 s_2) \frac{u-v+c}{c}
+(e_2 w_1 n_1 s_2+e_1 w_2 n_2 s_1) \nn \\
&+(e_1 w_1 n_2 s_2+e_2 w_2 n_1 s_1) \frac{u-v}{c},
\end{align}
and further rewriting in the following way
\begin{align}
&(e_1 w_1 n_1 s_1+e_2 w_2 n_2 s_2) \frac{u-v+c}{c}
+(e_2 w_1 n_1 s_2+e_1 w_2 n_2 s_1) \Bigg(
\frac{u-v+c}{c}-\frac{u-v}{c}
\Bigg)
\nn \\
&+(e_1 w_1 n_2 s_2+e_2 w_2 n_1 s_1) \frac{u-v}{c} \nn \\
=&(e_1 w_1 n_1 s_1+e_2 w_2 n_2 s_2+e_2 w_1 n_1 s_2+e_1 w_2 n_2 s_1) \frac{u-v+c}{c} \nn \\
&+(e_1 w_1 n_2 s_2+e_2 w_2 n_1 s_1-e_2 w_1 n_1 s_2-e_1 w_2 n_2 s_1) \frac{u-v}{c} \nn \\
=&(e_1 s_1+e_2 s_2)(n_1 w_1+n_2 w_2)
\frac{u-v+c}{c}
-(e_2 n_1-e_1 n_2)(s_2 w_1-s_1 w_2) \frac{u-v}{c} \nn \\
=&(e_1 s_1+e_2 s_2)(n_1 w_1+n_2 w_2) \frac{u-v}{c}
\Bigg(
\frac{u-v+c}{u-v} -\frac{(e_2 n_1-e_1 n_2)(s_2 w_1-s_1 w_2)}{(e_1 s_1+e_2 s_2)(n_1 w_1+n_2 w_2)}
\frac{c}{u-u+c}
\Bigg),
\end{align}
we can check this is the $m=n=1$ case of
\eqref{BPSdeterminant}.

\begin{figure}[h] 
\centering
\includegraphics[width=10cm]{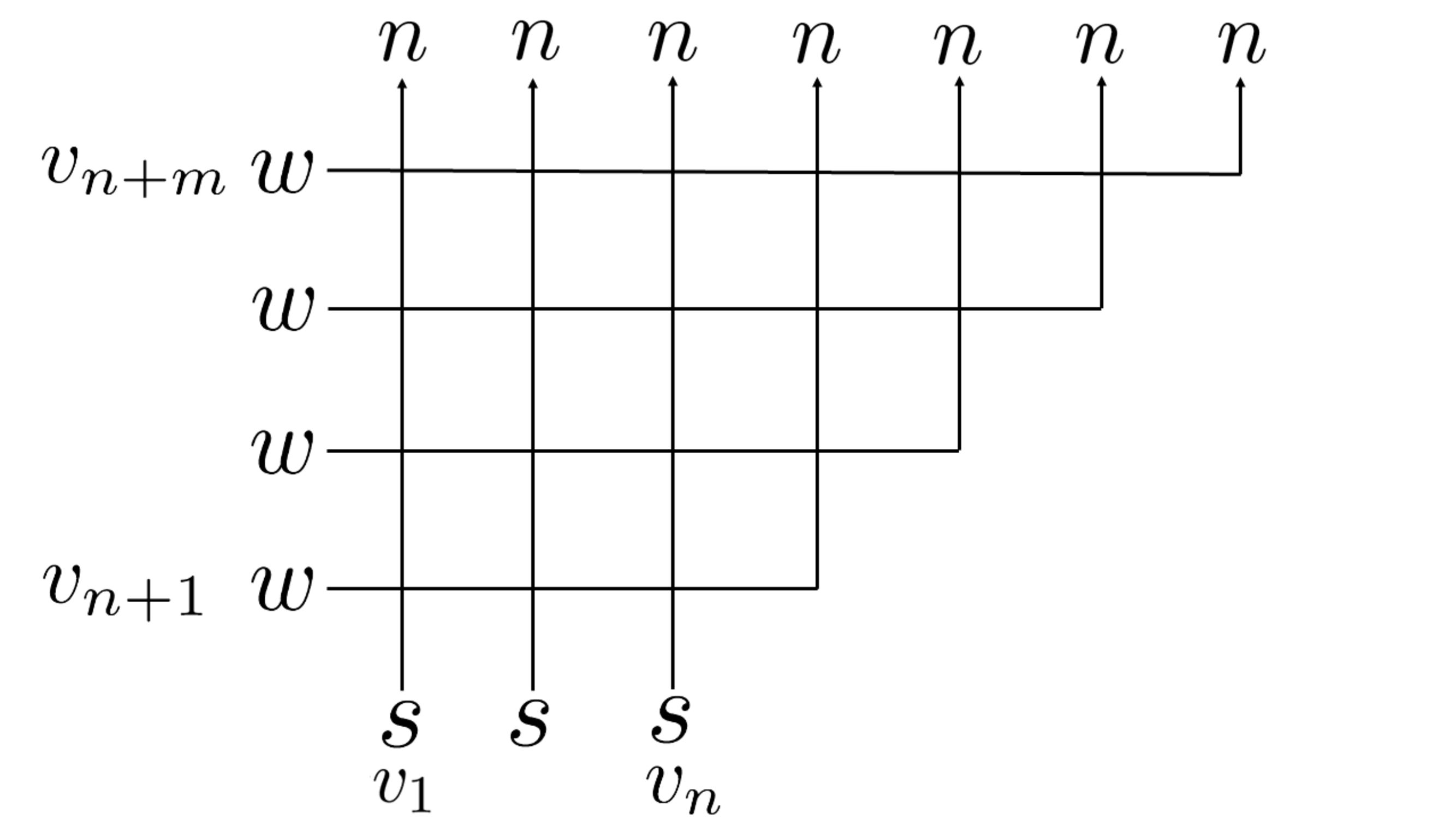}
\caption{Partition functions under trapezoid boundary
$T_{n,m}(v_1,\dots,v_{n}|v_{n+1},\dots,v_{n+m})$.}
\label{figuretrapezoidboundary}
\end{figure}

We use the following specialization later.
\begin{lemma}
We have
\begin{align}
&Z_{m,n}(u_1,\dots,u_m|v_1,\dots,v_n)|_{u_1=v_{n-m+1},u_2=v_{n-m+2},\dots,u_m=v_{n}} \nn \\
=&(n_1 s_1+n_2 s_2)^{n-m} (e_1 s_1+e_2 s_2)^m (n_1 w_1+n_2 w_2)^m
\prod_{i=n-m+1}^n \prod_{j=1}^n  \Bigg(\frac{v_i-v_j+c}{c} \Bigg). \label{BPSexpressionspecialization}
\end{align}
\end{lemma}
\begin{proof}
This can be checked in the same way with \cite[Prop 4.8]{MO}, for example.
Denote the rational function in the right hand side of \eqref{BPSexpression}
as $F_{m,n}(u_1,\dots,u_m|v_1,\dots,v_n) $.
It is easy to see that due to the factor $\displaystyle \prod_{\substack{ i \in K \\ 1 \le k \le n }}  (u_i-v_k)$,
only the summand corresponding to $K=\phi$ survives
after the specialization $u_1=v_{n-m+1},u_2=v_{n-m+2},\dots,u_m=v_{n}$,
and we get
\begin{align}
&F_{m,n}(u_1,\dots,u_m|v_1,\dots,v_n)|_{u_1=v_{n-m+1},u_2=v_{n-m+2},\dots,u_m=v_{n}} \nn \\
=&(n_1 s_1+n_2 s_2)^{n-m} (e_1 s_1+e_2 s_2)^m (n_1 w_1+n_2 w_2)^m
\prod_{i=n-m+1}^n \prod_{j=1}^n  \Bigg(\frac{v_i-v_j+c}{c} \Bigg),
\end{align} and hence
\eqref{BPSexpressionspecialization}.
\end{proof}

\begin{figure}[h] 
\centering
\includegraphics[width=10cm]{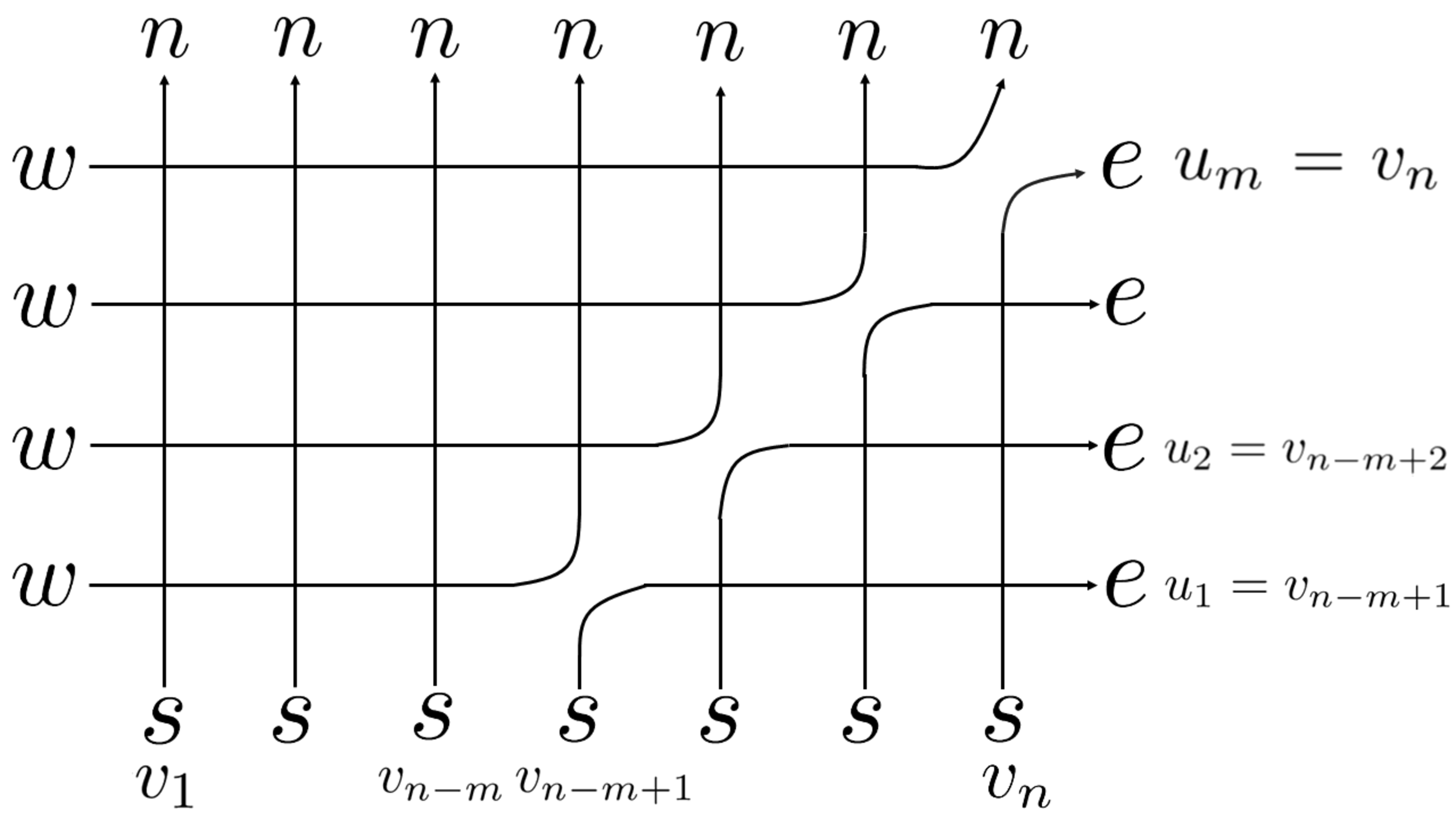}
\caption{Specialization $u_1=v_{n-m+1},u_2=v_{n-m+2},\dots,u_m=v_{n}$ of the generalized domain wall boundary partition fuctions
$Z_{m,n}(u_1,\dots,u_m|v_1,\dots,v_n)$.
From the $(n-m+1)$-th to the $n$-th columns, one $R$-matrix in each column turns into a permutation operator
and these permutation operators separate the original regime into two regimes.
The left one gives $T_{n-m,m}(v_1,\dots,v_{n-m}|v_{n-m+1},\dots,v_n)$
and the right one gives $Z_m(v_{n-m+1},\dots,v_n)$,
and taking product gives an expression for the specialization.
}
\label{figurerdomainwallspecialization}
\end{figure}

Now we introduce
partition functions under trapezoid boundary
as Figure \ref{figuretrapezoidboundary}
and denote as
$T_{n,m}(v_1,\dots,v_{n}|v_{n+1},\dots,v_{n+m})$.

We have the following factorized expressions
for $T_{n,m}(v_1,\dots,v_{n}|v_{n+1},\dots,v_{n+m})$.

\begin{theorem}
We have
\begin{align}
&T_{n,m}(v_1,\dots,v_{n}|v_{n+1},\dots,v_{n+m}) \nn \\
=&(n_1 s_1+n_2 s_2)^{n} (n_1 w_1+n_2 w_2)^m \nn \\
\times&\prod_{i=n+1}^{n+m} \prod_{j=1}^{n} \Bigg( \frac{v_i-v_j+c}{c}  \Bigg)
\prod_{n+1 \le i < j \le n+m} \Bigg( \frac{v_j-v_i+c}{c} \Bigg).
\label{factorizationtrapezoid}
\end{align}
\end{theorem}
\begin{proof}

By relabeling, \eqref{factorizationtrapezoid} can be rewritten as
\begin{align}
&T_{n-m,m}(v_1,\dots,v_{n-m}|v_{n-m+1},\dots,v_n) \nn \\
=&(n_1 s_1+n_2 s_2)^{n-m} (n_1 w_1+n_2 w_2)^m \nn \\
\times&\prod_{i=n-m+1}^n \prod_{j=1}^{n-m} \Bigg( \frac{v_i-v_j+c}{c}  \Bigg)
\prod_{n-m+1 \le i < j \le n} \Bigg( \frac{v_j-v_i+c}{c} \Bigg).
\label{factorizationtrapezoidtwo}
\end{align}
We show \eqref{factorizationtrapezoidtwo}
by combining the results \eqref{BPSexpressionspecialization} and
\eqref{twopartfactorization} together with the following observation.
At the level of partition functions, we find the specialization $u_1=v_{n-m+1},u_2=v_{n-m+2},\dots,u_m=v_{n}$
factorizes the generalized domain wall boundary partition functions as
\begin{align}
&Z_{m,n}(u_1,\dots,u_m|v_1,\dots,v_n)|_{u_1=v_{n-m+1},u_2=v_{n-m+2},\dots,u_m=v_{n}}  \nn \\
=&
T_{n-m,m}(v_1,\dots,v_{n-m}|v_{n-m+1},\dots,v_n)
Z_m(v_{n-m+1},\dots,v_n). \label{twopartfactorization}
\end{align}
This factorization can be understood from the graphical description
of $Z_{m,n}(u_1,\dots,u_m|v_1,\dots,v_n)$ as follows.
One finds that from the $(n-m+1)$-th to the $n$-th columns,
the specialization $u_1=v_{n-m+1},u_2=v_{n-m+2},\dots,u_m=v_{n}$
turns one $R$-matrix in each column to a permutation operator, and as graphically represented in Figure \ref{figurerdomainwallspecialization},
one notes this specialization separates the original regime into two regimes.
The partitition function for the left regime gives
$T_{n-m,m}(v_1,\dots,v_{n-m}|v_{n-m+1},\dots,v_n)$
and the one for the right regime gives $Z_m(v_{n-m+1},\dots,v_n)$,
and the product of these partition functions give \\
$Z_{m,n}(u_1,\dots,u_m|v_1,\dots,v_n)|_{u_1=v_{n-m+1},u_2=v_{n-m+2},\dots,u_m=v_{n}}$,
i.e. we get \eqref{twopartfactorization}.

Combining \eqref{factorizationhalftwist},
\eqref{BPSexpressionspecialization} and
\eqref{twopartfactorization}, we get
\begin{align}
&T_{n-m,m}(v_1,\dots,v_{n-m}|v_{n-m+1},\dots,v_n)
=Z_m(v_{n-m+1},\dots,v_n)^{-1} \nn \\
\times&Z_{m,n}(u_1,\dots,u_m|v_1,\dots,v_n)|_{u_1=v_{n-m+1},u_2=v_{n-m+2},\dots,u_m=v_{n}}  \nn \\
=&
\frac{1}{(e_1 s_1+e_2 s_2)^m} \prod_{n-m+1 \le i < j \le n} \Bigg( \frac{c}{v_i-v_j+c} \Bigg) \nn \\
\times&(n_1 s_1+n_2 s_2)^{n-m} (e_1 s_1+e_2 s_2)^m (n_1 w_1+n_2 w_2)^m
\prod_{i=n-m+1}^n \prod_{j=1}^n  \Bigg(\frac{v_i-v_j+c}{c} \Bigg),
\end{align}
and simplification gives \eqref{factorizationtrapezoidtwo}.

\end{proof}

\begin{figure}[h] 
\centering
\includegraphics[width=10cm]{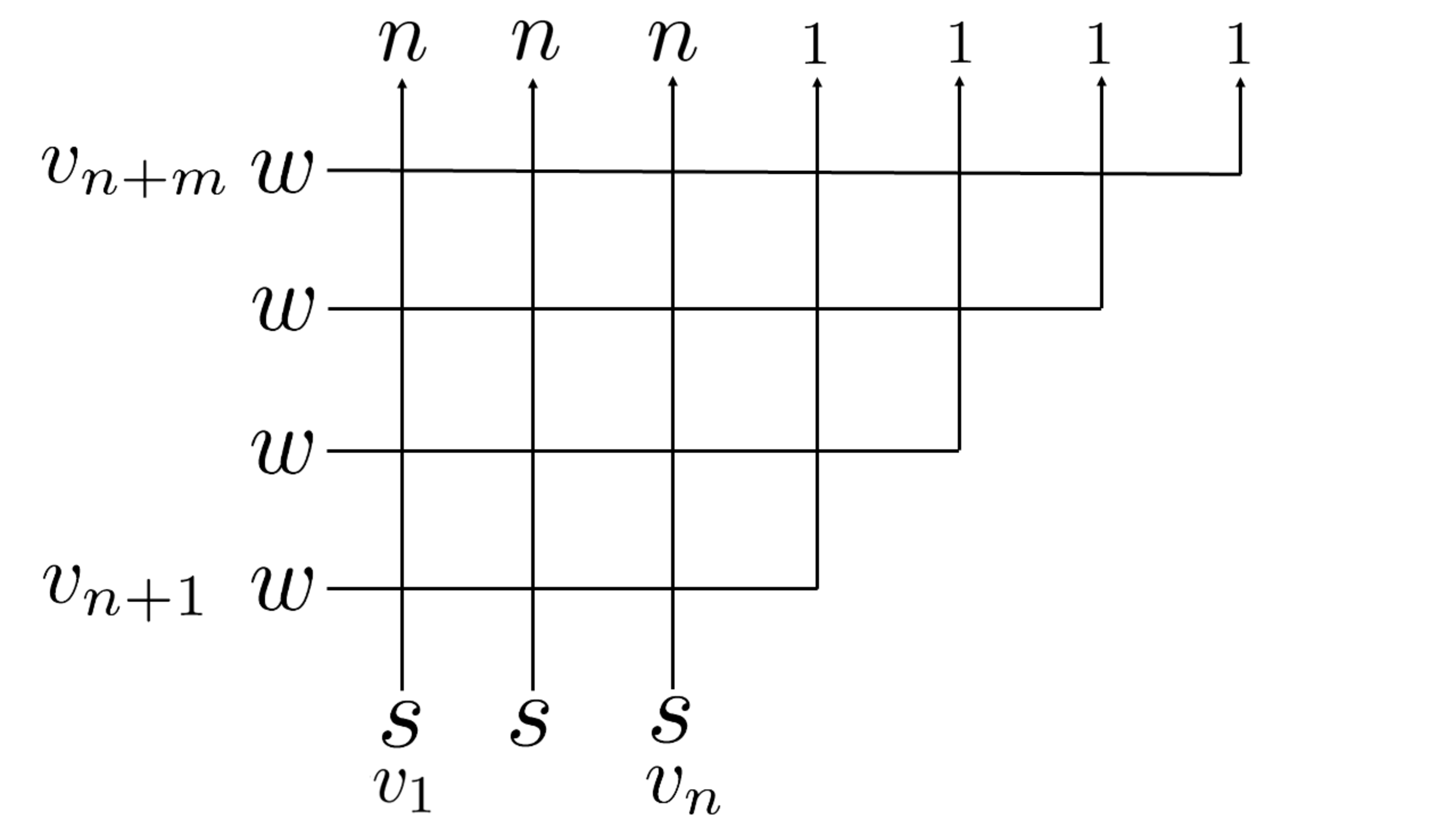}
\caption{Emptiness formation probability $P(m)$.}
\label{figuretrapezoidEFP}
\end{figure}

\begin{figure}[h] 
\centering
\includegraphics[width=10cm]{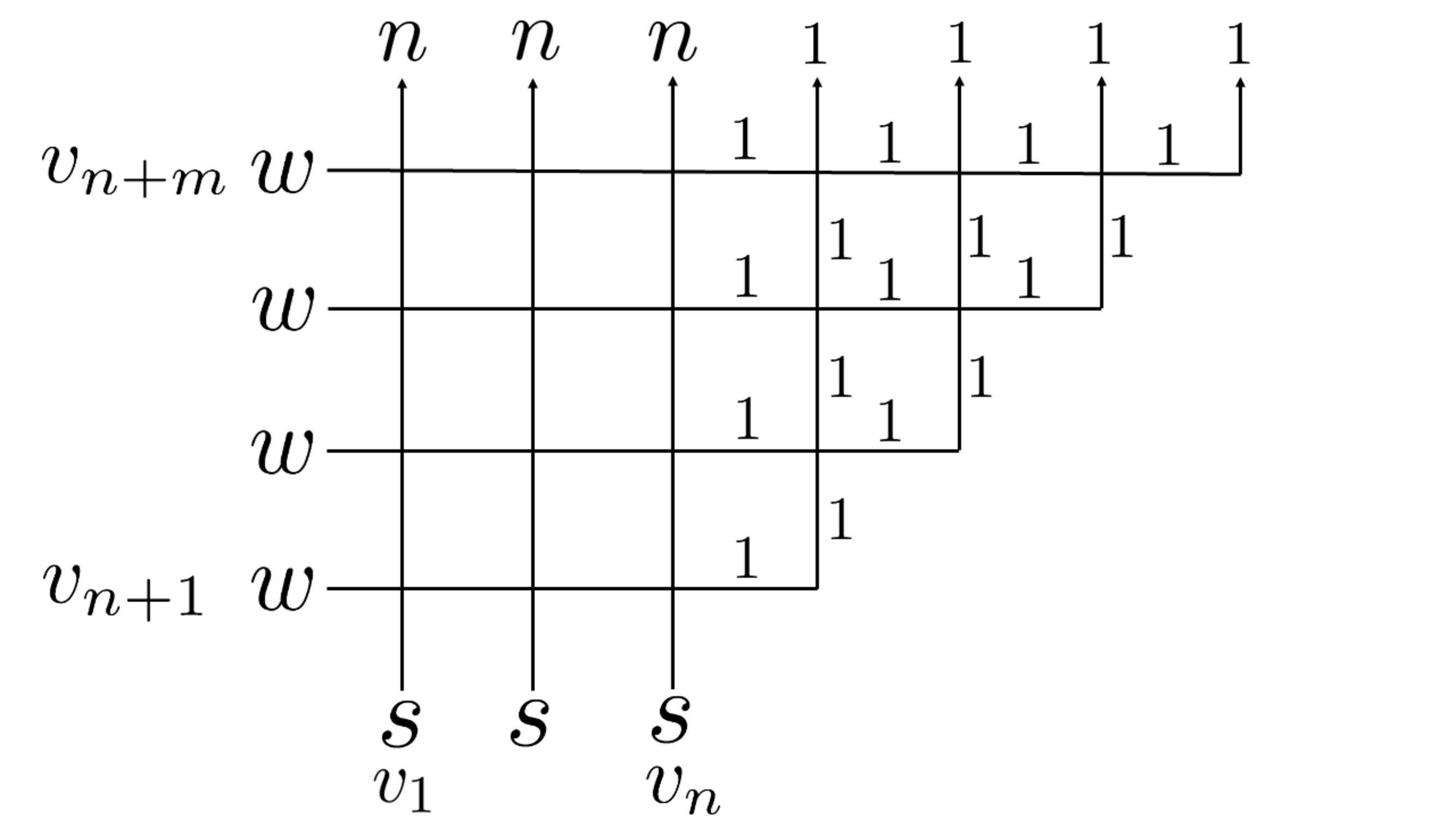}
\caption{Emptiness formation probability $P(m)$ is equivalent to the probability
that the states in the triangular regime are all fixed to be 1.
}
\label{figuretrapezoidEFPfreeze}
\end{figure}

Finally, we discuss a class of probabilities which admits determinant expressions.
Here we assume $v_1 \le v_2 \le \cdots \le v_{n+m}$, $n_1,n_2,s_1,s_2,w_1,w_2 \ge 0$
for nonnegativity of the partition functions and quantities which we introduce.
Let us introduce a version of emptiness formation probability
which the states of $m$-consecutive part from right in the upper boundary
are  all fixed to be 1, as graphically represented  in Figure \ref{figuretrapezoidEFP}.
We denote the probability as $P(m)$.
We can show the following.
\begin{theorem}
The emptiness formation probability $P(m)$ has the following determinant representation
\begin{align}
P(m)
=&\Bigg( \frac{n_1 s_1}{n_1 s_1+n_2 s_2} \Bigg)^m 
\nn \\
\times& \det_{1 \le j,k \le m}
\Bigg( \delta_{jk} 
-\frac{\gamma c}{v_{j+n}-v_{k+n}+c} 
\prod_{i=1}^n \frac{v_{j+n}-v_i}{v_{j+n}-v_i+c}
\prod_{ \substack{i=1 \\ i \neq j}  }^m
\frac{v_{j+n}-v_{i+n}+c}{v_{j+n}-v_{i+n}}
\Bigg), \label{EFP}
\end{align}
where
$\displaystyle \gamma=\frac{-n_2 (s_2 w_1-s_1 w_2)}{s_1(n_1 w_1+n_2 w_2)}$.
\end{theorem}
\begin{proof}
One can see that all the states in the triangular regime are frozen to be 1 as given in Figure \ref{figuretrapezoidEFPfreeze},
once the  states of $m$-consecutive part from right in the top boundary
are all fixed to be 1.
Then we find $P(m)$ is given by

\begin{align}
P(m)=&
Z_{m,n}(v_{n+1},\dots,v_{n+m}|v_1,\dots,v_n)|_{e_1=1,e_2=0}
\times
n_1^m \prod_{n+1 \le i < j \le n+m} \Bigg( \frac{v_j-v_i+c}{c} \Bigg) \nn \\
\times&T_{n,m}(v_1,\dots,v_{n}|v_{n+1},\dots,v_{n+m})^{-1}.
\label{EFPfactorization}
\end{align}
Note $Z_{m,n}(v_{n+1},\dots,v_{n+m}|v_1,\dots,v_n)|_{e_1=1,e_2=0}$ comes from the rectangular regime,
the factor
$\displaystyle \prod_{n+1 \le i < j \le n+m} \Bigg( \frac{v_j-v_i+c}{c} \Bigg)$ comes from the
$R$-matrix elements in the triangular regime,
$n_1^m$ comes from taking $m$ consecutive states from right in the top boundary to be 1,
and the partition function $T_{n,m}(v_1,\dots,v_{n}|v_{n+1},\dots,v_{n+m})$ is the normalization factor.

Inserting into \eqref{EFPfactorization}
the factorized expression for $T_{n,m}(v_1,\dots,v_{n}|v_{n+1},\dots,v_{n+m})$
\eqref{factorizationtrapezoid}
and
\begin{align}
&Z_{m,n}(v_{n+1},\dots,v_{n+m}|v_1,\dots,v_n)|_{e_1=1,e_2=0} \nn \\
=&s_1^m (n_1 s_1+n_2 s_2)^{n-m} (n_1 w_1+n_2 w_2)^m
\prod_{\substack{1 \le i \le m \\ 1 \le j \le n}}
\Bigg( \frac{v_{i+n}-v_j}{c} \Bigg)
\nn \\
\times& \det_{1 \le j,k \le m}
\Bigg( \delta_{jk} \prod_{i=1}^n \frac{v_{j+n}-v_i+c}{v_{j+n}-v_i}
-\frac{\gamma c}{v_{j+n}-v_{k+n}+c} \prod_{ \substack{i=1 \\ i \neq j}  }^m \frac{v_{j+n}-v_{i+n}+c}{v_{j+n}-v_{i+n}}
\Bigg),
\end{align}
which follows from \eqref{BPSdeterminant}
and after simplification, we get \eqref{EFP}.

\end{proof}

\section{Conclusion}
We showed that for the rational six vertex model, partition functions
under triangular boundary and its extension to trapezoid boundary exhibit factorization property.
The boundary conditions are mixture of up and down spins and
are motivated and inspired by the recent work on
the generalized domain wall boundary partition functions \cite{BPS}.
Factorization means that when dealing with correlation fuctions,
denominator part becomes easier than the usual domain wall boundary.
As an illustration, we introduce a version of emptiness formation probabilities
which admits determinant representations.
It would be interesting for example to investigate to which extent this factorization property holds.

\section*{Acknowledgements}
The author thanks
Jan de Gier, Saburo Kakei, William Mead, Yuan Miao and Junji Suzuki
for discussions and comments.
This work is supported by Grant-in-Aid for Scientific Research 21K03176, JSPS.

\end{document}